\documentclass[11pt]{article}
\usepackage{geometry}
\usepackage{graphicx}
\usepackage{amsmath}
\usepackage{amssymb}
\usepackage[all]{xy}
\usepackage[tight]{subfigure}
\usepackage{xspace}

\newtheorem{theorem}{Theorem}[section]
\newtheorem{corollary}[theorem]{Corollary}
\newtheorem{lemma}[theorem]{Lemma}
\newtheorem{proposition}[theorem]{Proposition}
\newtheorem{claim}[theorem]{Claim}

\newtheorem{example}[theorem]{Example}
\newtheorem{observation}[theorem]{Observation}

\def\squarebox#1{\hbox to #1{\hfill\vbox to #1{\vfill}}}
\newcommand{\qed}{\hspace*{\fill}\vbox{\hrule\hbox{\vrule\squarebox{.667em}\vrule}\hrule}\smallskip}
\newenvironment{proof}{\noindent{\bf Proof:~~}}{\(\qed\)}

\newcommand{\diag}{d}
\newcommand{\io}{s} 
\newcommand{\no}{x}    
\newcommand{\oo}{y}    
\newcommand{\co}{c}    
\newcommand{\ao}{z}    
\newcommand{\oi}{v}    
\newcommand{\poa}{\mathrm{PoA}}
\newcommand{\game}{opinion game }

 \advance \topmargin by -\headheight
 \advance \topmargin by -\headsep
 
 \evensidemargin \oddsidemargin
\def\coa{\tilde{c}}

 \newcommand{\xhdr}[1]{\subsubsection*{{\bf #1}}}

\newcommand{\omt}[1]{}
\newcommand{\Xomit}[1]{}
\newcommand{\supproof}[1]{}
\newcommand{\supproofeq}[1]{}

\newlength{\saveparindent}
\newlength{\saveparskip}
{\end{list}}

\begin{document}

%
%
%
%

\title{
How Bad is Forming Your Own Opinion?
\thanks{
Supported in part by 
the MacArthur Foundation,
the Sloan Foundation, 
a Google Research Grant,
a Yahoo!~Research Alliance Grant,
and NSF grants 
IIS-0910664, 
CCF-0910940, 
and IIS-1016099. 
}
}

\author{
David Bindel
\thanks{
Department of Computer Science,
Cornell University, Ithaca NY 14853.
Email: bindel@cs.cornell.edu.
}
\and
Jon Kleinberg
\thanks{
Department of Computer Science,
Cornell University, Ithaca NY 14853.
Email: kleinber@cs.cornell.edu.
}
 \and 
Sigal Oren
\thanks{
Department of Computer Science,
Cornell University, Ithaca NY 14853.
Email: sigal@cs.cornell.edu.
}
}

\begin{titlepage}

\maketitle

\begin{abstract}


The question of how people form their opinion
has fascinated economists and sociologists for quite some time.
In many of the models, a group of people in a social network, each holding
a numerical opinion, arrive at a shared opinion through
repeated averaging with their neighbors in the network.
Motivated by the observation that consensus is rarely reached
in real opinion dynamics, we study a related sociological model
in which individuals' intrinsic beliefs counterbalance 
the averaging process and yield a diversity of opinions.

By interpreting the repeated averaging as best-response dynamics
in an underlying game with natural payoffs, and the limit of
the process as an equilibrium, we are able to study the
cost of disagreement in these models relative to a social optimum.  
We provide a tight bound on the cost at equilibrium relative 
to the optimum; our analysis draws a connection between these
agreement models and extremal problems that lead to generalized eigenvalues.
We also consider a natural network design problem in this setting:
which links can we add to the underlying network to reduce the cost of
disagreement at equilibrium?


\end{abstract}

\thispagestyle{empty}
\end{titlepage}

\section{Introduction}
\label{sec:intro}
\xhdr{Averaging Opinions in a Social Network}
An active line of recent work in economic theory has considered
processes by which a group of people in a social network
can arrive at a shared opinion through a form of repeated averaging
\cite{acemoglu-misinformation,demarzo-opinions,golub-network-learning,jackson-networks-book}.
This work builds on a basic model of DeGroot \cite{degroot-opinions},
in which we imagine that each person $i$ holds an {\em opinion} 
equal to a real number
$\ao_i$, which might for example represent a position on a 
political spectrum, or a probability that $i$ assigns to a certain belief.
There is a weighted graph $G = (V,E)$ representing a social network, and 
node $i$ is influenced by the opinions of her neighbors in $G$, with
the edge weights reflecting the extent of this influence.
Now, in each time step node $i$ updates her opinion to be a weighted average
of her current opinion and the current opinions of her neighbors.

This body of work has developed a set of general conditions under
which such processes will converge to a state of {\em consensus}, 
in which all nodes hold the same opinion.  This emphasis on consensus,
however, can only model a
specific type of opinion dynamics, where the opinions of the group 
all come together.  
As the sociologist David Krackhardt has observed,
\begin{quote}
{\footnotesize
We should not ignore the fact 
that in the real world consensus is usually not 
reached. Recognizing this, most traditional 
social network scientists do not focus on an 
equilibrium of consensus. They are instead 
more likely to be concerned with explaining the 
lack of consensus (the variance) in beliefs and 
attitudes that appears in actual social influence
contexts \cite{krackhardt-rev-jackson}. }
\end{quote}

In this paper we study a model of opinion dynamics 
in which consensus is not reached in general, with
the goal of quantifing the inherent social cost of this 
lack of consensus.
To do this, we first need a framework that captures some of the underlying
reasons why consensus is not reached, as well as a way of measuring the
cost of disagreement.

\xhdr{Lack of Agreement and its Cost}
We begin from a variation on the DeGroot model due to 
Friedkin and Johnsen \cite{friedkin-initial-opinions},
which posits that each node $i$ maintains a persistent 
{\em internal opinion} $\io_i$. This internal opinion remains constant even
as node $i$ updates her overall opinion $\ao_i$ through averaging.
More precisely, if $w_{i,j} \geq 0$ denotes the weight on the
edge $(i,j)$ in $G$, then in one time step node $i$ updates
her opinion to be the average
\begin{equation}
\label{eq:avg-with-io}
\ao_i = 
\dfrac{\io_i + \sum_{j \in N(i)} w_{i,j} \ao_j}{1 + \sum_{j \in N(i)} w_{i,j}},
\end{equation}
where $N(i)$ denotes the set of neighbors of $i$ in $G$.
Note that, in general, the presence of $\io_i$ as a constant in each 
iteration prevents repeated averaging from bringing all nodes to the same opinion.
In this way, the model distinguishes between an individual's
intrinsic belief $\io_i$ and her overall opinion $\ao_i$;
the latter represents a compromise between the persistent value of $\io_i$ 
and the expressed opinions of others to whom $i$ is connected.
This distinction between $\io_i$ and $\ao_i$ also has parallels 
in empirical work that seeks to trace deeply held opinions such 
as political orientations back to differences in education and
background, and even to explore genetic bases for such
patterns of variation \cite{alford-opinion-genetic}.

Now, if consensus is not reached, how should we quantify the cost
of this lack of consensus?
Here we observe that since the standard models use averaging as
their basic mechanism, we can equivalently view nodes' actions in
each time step as myopically optimizing a quadratic cost function:
Updating $\ao_i$ as in Equation (\ref{eq:avg-with-io}) is 
the same as choosing $\ao_i$ to minimize 
\begin{equation}
\label{eq:ao-cost}
(\ao_i-\io_i)^2+\sum_{j \in N(i)} w_{i,j}(\ao_i-\ao_j)^2.
\end{equation}
We therefore take this as the {\em cost} that $i$ incurs by
choosing a given value of $\ao_i$, so that averaging becomes
a form of cost minimization.

Given this view, we can think of repeated averaging 
as the trajectory of 
best-response dynamics in a one-shot, complete information game 
played by the nodes in $V$,
where $i$'s strategy is a choice of opinion $\ao_i$, and her payoff
is the negative of the cost in Equation (\ref{eq:ao-cost}).

\xhdr{Nash Equilibrium and Social Optimality in a Game of Opinion Formation}
In this model, repeated averaging does converge to the
unique Nash equilibrium of the game defined by the individual cost functions
in (\ref{eq:ao-cost}): 
each node $i$ has an opinion $\no_i$ that is the weighted
average of $i$'s internal opinion and the (equilibrium) opinions of 
$i$'s neighbors.
This equilibrium will not in general correspond to the 
{\em social optimum}, the vector of node opinions 
$y$ that minimizes the {\em social cost},
defined to be sum of all players' costs:
$\co(\oo)=\sum_i \left( (\oo_i-\io_i)^2
+\sum_{j \in N(i)} w_{i,j}(\oo_i-\oo_j)^2 \right).$

The sub-optimality of the Nash equilibrium can be viewed in
terms of the {\em externality} created by a player's 
personal optimization: by refusing to move further toward
their neighbors' opinions, players can cause additional cost
to be incurred by these neighbors.
In fact we can view the problem of minimizing
social cost for this game as a type of 
{\em metric labeling problem} 
\cite{boykov-fast-apx-min,kleinberg-met-label},
albeit a polynomial-time solvable case of the problem
with a non-metric quadratic distance function on the real numbers:
we seek node labels that balance the value of a 
cost function at each node (capturing disagreement with
node-level information) and a cost function for label
disagreement across edges.
Viewed this way, the sub-optimality of Nash equilibrium becomes
a kind of sub-optimality for local optimization.

A natural question for this game is thus the {\em price of anarchy},
defined as the ratio between the cost of the Nash
equilibrium and the cost of the optimal solution.

\xhdr{Our Results: Undirected Graphs}
The model we have described can be used as stated in 
both undirected and directed graphs --- the only difference is
in whether $i$'s neighbor set $N(i)$ represents the nodes with whom 
$i$ is connected by undirected edges, or to whom $i$ links
with directed edges.
However, the behavior of the price of anarchy is very different
in undirected and directed graphs, and so we analyze them separately,
beginning with the undirected case.

As an example of how a sub-optimal social cost can arise at
equilibrium in an undirected graph, consider
the graph depicted in Figure~\ref{fig:no-compromise-example} --- 
a three-node path in which the nodes have internal opinions
$0$, $1/2$, and $1$ respectively.
As shown in the figure,
the ratio between the social cost of the Nash equilibrium and
the social optimum is $9/8$.
Intuitively, the reason for the higher cost of the Nash equilibrium
is that the center node --- by symmetry --- cannot usefully shift
her opinion in either direction, and so to achieve optimality the two outer
nodes need to compromise more than they want to at equilibrium.
This is a reflection of the externality discussed above,
and it is the qualitative source of sub-optimality in general for equilibrium
opinions --- nodes move in the direction of their neighbors,
but not sufficiently to achieve the globally minimum social cost.

\begin{figure}
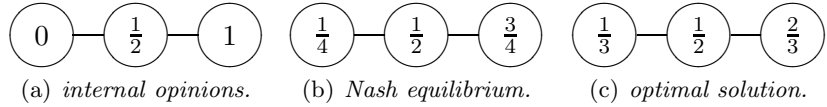

\begin{center}
\subfigure[\emph{internal opinions.}]{
\xygraph{ !{<0cm,0cm>;<1.25cm,0cm>:<0cm,1.25cm>::} !{(1,0) }*+[o]++[F]{0}="1" !{(2,0) }*+[o]++[F]{\frac{1}{2}}="2" !{(3,0) }*+[o]++[F]{1}="3" "1"-"2" "2"-"3"} 
 \label{fig:init-op}
}
\subfigure[\emph{Nash equilibrium.}]{
\xygraph{ !{<0cm,0cm>;<1.25cm,0cm>:<0cm,1.25cm>::} !{(1,0) }*+[o]++[F]{\frac{1}{4}}="1" !{(2,0) }*+[o]++[F]{\frac{1}{2}}="2" !{(3,0) }*+[o]++[F]{\frac{3}{4}}="3" "1"-"2" "2"-"3"}
 \label{fig:nash-op}
}
\subfigure[\emph{optimal solution.}]{
\xygraph{ !{<0cm,0cm>;<1.25cm,0cm>:<0cm,1.25cm>::} !{(1,0) }*+[o]++[F]{\frac{1}{3}}="1" !{(2,0) }*+[o]++[F]{\frac{1}{2}}="2" !{(3,0) }*+[o]++[F]{\frac{2}{3}}="3" "1"-"2" "2"-"3"}
 \label{fig:optimal-op}
}
\caption{
{\small
An example in which the two players on the sides 
do not compromise by the optimal amount, given that
the player in the middle should not shift her opinion.
The social cost of the optimal set of opinions is $1/3$,
while the cost of the Nash equilibrium is $3/8$.
\label{fig:no-compromise-example}
}
}
\end{center}
\vspace*{-0.2in}
\end{figure}

Our first result is that the very simple example in 
Figure~\ref{fig:no-compromise-example} is in fact extremal
for undirected graphs: we show that for any undirected graph $G$ and any
internal opinions vector $s$, the price of anarchy is at most $9/8$.
We prove this by casting the question as an extremal problem
for quadratic forms, and analyzing the resulting structure using 
eigenvalues of the Laplacian matrix of $G$.
From this, we obtain a characterization of the set of graphs $G$
for which some internal opinions vector $\io$ yields a price of anarchy
of $9/8$.

We show that this bound of $9/8$ continues to hold even for
some generalizations of the model --- when nodes $i$ have different
coefficients $w_i$ on the cost terms for their internal opinions,
and when certain nodes
are ``fixed'' and simply do not modify their opinions.

\begin{figure}
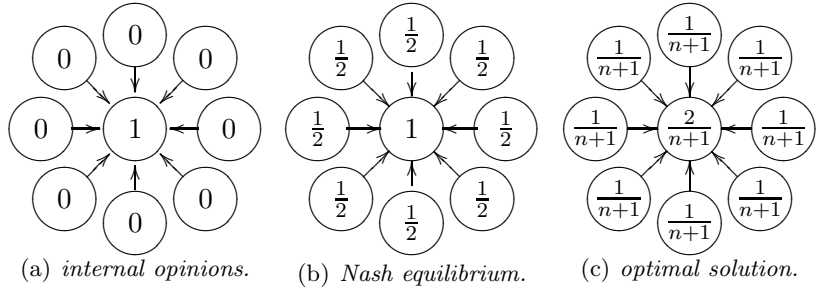

\begin{center}
\subfigure[\emph{internal opinions.}]{
\xygraph{ !{<0cm,0cm>;<1.25cm,0cm>:<0cm,1.25cm>::} !{(0.25,0.25) }*+[o]+<12pt>[F]{0}="1" !{(0,1) }*+[o]+<12pt>[F]{0}="2" !{(0.25,1.75) }*+[o]+<12pt>[F]{0}="3" !{(1,0) }*+[o]+<12pt>[F]{0}="4"
!{(1,1) }*+[o]++<9pt>[F]{1}="5" !{(1,2) }*+[o]+<12pt>[F]{0}="6" !{(1.75,0.25) }*+[o]+<12pt>[F]{0}="7"  !{(2,1) }*+[o]+<12pt>[F]{0}="8"  !{(1.75,1.75) }*+[o]+<12pt>[F]{0}="9" "1":"5" "2":"5" "3":"5" "4":"5" "6":"5" "7":"5" "8":"5" 
"9":"5" } 
 \label{fig:init-op-dir}
}
\subfigure[\emph{Nash equilibrium.}]{
  \xygraph{ !{<0cm,0cm>;<1.25cm,0cm>:<0cm,1.25cm>::} !{(0.25,0.25) }*+[o]+<9pt>[F]{\frac{1}{2}}="1" !{(0,1) }*+[o]+<9pt>[F]{\frac{1}{2}}="2" !{(0.25,1.75) }*+[o]+<9pt>[F]{\frac{1}{2}}="3" !{(1,0) }*+[o]+<9pt>[F]{\frac{1}{2}}="4"
!{(1,1) }*+[o]+<12pt>[F]{1}="5" !{(1,2) }*+[o]+<9pt>[F]{\frac{1}{2}}="6" !{(1.75,0.25) }*+[o]+<9pt>[F]{\frac{1}{2}}="7"  !{(2,1) }*+[o]+<9pt>[F]{\frac{1}{2}}="8"  !{(1.75,1.75) }*+[o]+<9pt>[F]{\frac{1}{2}}="9" "1":"5" "2":"5" 
"3":"5" "4":"5" "6":"5" "7":"5" "8":"5" "9":"5" } 
 \label{fig:nash-op-dir}
}
\subfigure[\emph{optimal solution.}]{
 \xygraph{ !{<0cm,0cm>;<1.25cm,0cm>:<0cm,1.25cm>::} !{(0.25,0.25) }*+[o]+=[F]{\frac{1}{n+1}}="1" !{(0,1) }*+[o]+=[F]{\frac{1}{n+1}}="2" !{(0.25,1.75) }*+[o]+=[F]{\frac{1}{n+1}}="3" !
{(1,0) }*+[o]+=[F]{\frac{1}{n+1}}="4"
!{(1,1) }*+[o]+=[F]{\frac{2}{n+1}}="5" !{(1,2) }*+[o]+=[F]{\frac{1}{n+1}}="6" !{(1.75,0.25) }*+[o]+=[F]{\frac{1}{n+1}}="7"  !{(2,1) }*+[o]+=[F]{\frac{1}{n+1}}="8"  !{(1.75,1.75) }*+[o]+=[F]{\frac{1}{n
+1}}="9" "1":"5" "2":"5" "3":"5" "4":"5" "6":"5" "7":"5" "8":"5" "9":"5" } 
 \label{fig:optimal-op-dir}
}
\caption{
{\small
An example demonstrating that the price of anarchy 
of a directed graph can be unbounded.
\label{fig:star-example}
}
}
\end{center}
\vspace*{-0.2in}
\end{figure}

\xhdr{Our Results: Directed Graphs}
We next consider the case in which $G$ is a directed graph; 
the form of the cost functions remains exactly the same, with
directed edges playing the role of undirected ones, but
the range of possible behaviors in the model becomes very different.
This is due to the fact that nodes can now exert a large influence over the
network without being influenced themselves.
Indeed, as Matt Jackson has observed, directed versions of repeated averaging
models can naturally incorporate ``external'' media sources;
we simply include nodes with no outgoing links, so
that they maintain their internal opinion
\cite{jackson-networks-book}.

We first show that the spectral machinery developed for 
analyzing undirected graphs can be extended to the directed case;
through an approach based on generalized eigenvalue problems we can
efficiently compute the maximum possible price of anarchy, 
over all choices of internal node opinions, on a given graph $G$.
However, in contrast to the case of undirected graphs, 
the price of anarchy can be very large in some instances;
the simple example in Figure~\ref{fig:star-example} 
shows a case in which $n-1$ nodes with internal opinion $0$
all link to a single node that has
internal opinion $1$ and no out-going edges,
producing an in-directed star.
As a result, the social cost of the Nash equilibrium is $\frac{1}{2}(n-1)$,
whereas the minimum social cost is at most $1$, since 
the player at the center of
the star could simply shift her opinion to $0$.
Intuitively, this corresponds to a type of social network in 
which the whole group pays attention to a single influential
``leader'' or ``celebrity''; this drags people's opinions
far from their internal opinions $s_i$, creating a large social cost.
Unfortunately, the leader 
is essentially unaware of the people paying attention to her,
and hence has no incentive to modify her opinion in a direction
that could greatly reduce the social cost.

In Section~\ref{sec:directed} we show that a price of anarchy 
lower-bounded by a polynomial in $n$ 
can in fact be achieved in directed
graphs of constant degree, so this behavior is not simply a consequence
of large in-degree.
It thus becomes a basic question whether there are natural classes
of directed graphs, and even bounded-degree directed graphs,
for which a constant price of anarchy is achievable.

Unweighted Eulerian directed graphs are a natural class to consider ---
first, because they generalize undirected graphs,
and second, because they capture the idea that at least at a local level
no node has an asymmetric effect on the system.
We use our framework for directed graphs to 
derive two bounds on the price of anarchy of Eulerian graphs:
For Eulerian graphs with maximum degree $\Delta$  
we obtain a bound of $\Delta+1$ on the price of anarchy. For the
subclass of Eulerian {\em asymmetric} directed
graphs\footnote{An Eulerian {\em asymmetric} directed
graph is an Eulerian graph that does not contain any pair of 
oppositely oriented edges $(i,j)$ and $(j,i)$.}
with maximum degree $\Delta$ and edge expansion $\alpha$, 
we show a bound of $O(\Delta^2 \alpha^{-2})$ on the price of anarchy.


\xhdr{Our Results: Modifying the Network}
Finally, we consider an algorithmic problem within this framework 
of opinion formation. The question is the following: if we 
have the ability to modify 
the edges in the network (subject to certain
constraints), how should we do this to reduce the social cost
of the Nash equilibrium by as much as possible?
This is a natural question both as a self-contained issue within
the mathematical framework of opinion formation, and also 
in the context of applications: many social media sites overtly
and algorithmically consider how to balance the mix of news content
\cite{agarwal-online-content,backstrom-kdd09,
munson-diverse-political,munson-sidelines} and
also the mix of social content 
\cite{backstrom-icwsm11,sun-page-fanning}
that they expose their users to, 
so as to optimize user engagement on the site.

Adding edges to reduce the social cost has an intuitive basis:
it seems natural that exposing people to others with different opinions can 
reduce the extent of disagreement within the group.
When one looks at the form of the social cost $\co(\oo)$, however,
there is something slightly counter-intuitive about the idea of
adding edges to improve the situation: the social cost is a sum of
quadratic terms, and by adding edges to $G$ we are simply adding 
further quadratic terms to the cost.  For this reason, in fact, adding edges
to $G$ can never improve the optimal social cost.
But adding edges {\em can} improve the social cost of the Nash equilibrium,
and sometimes by a significant amount --- the point is that 
adding terms to the cost function shifts the equilibrium itself,
which can sometimes more than offset the additional terms.
For example, if we add a single edge from the center of the star
in Figure~\ref{fig:star-example} to one of the leaves,
then the center will shift her opinion to $2/3$ in equilibrium,
causing all the leaves to shift their opinions to $1/3$, and
resulting in a $\Theta(n)$ improvement in the social cost.
In this case, once the leader pays attention to
even a single member of the group, the social cost improves dramatically.

We focus on three main variants on this question: when all edges
must be added {\em to} a specific node (as in the case when a site
can modify the amount of attention directed to a media source or
celebrity); when all edges must be added {\em from} a specific node
(as in the case when a particular media site tries to shift its
location in the space of opinions by blending in content from others);
and when edges can be added between any pair of nodes in the network
(as in the case when a social networking site evaluates modifications
to its feeds of content from one user to another
\cite{backstrom-icwsm11,sun-page-fanning}).

In Section~\ref{sec:adding-edges} we show that, in the previously discussed variants,
the problem of where to add edges to optimally reduce the 
social cost is NP-hard. 
But we obtain a set of positive
results as well, including a $\frac94$-approximation algorithm when
edges can be added between arbitrary pairs of nodes, and
an algorithm to find the optimal amount of additional weight to add to
a given edge.

\section{Undirected Graphs} \label{sec:undirected}
We first consider the case of undirected graphs and later handle 
the more general case of directed graphs. 
The main result in this section is a tight bound on the price
of anarchy for the opinion-formation game in undirected graphs.
After this, we 
discuss
two slight extensions to the model:
in the first, each player can put a different amount of weight 
on her internal opinion; and in the second,
each player has several fixed opinions she listens to
instead of an internal opinion. 
We show that both models can be reduced to the basic
form of the model which we study first.

For undirected graphs we can simplify the social cost to 
the following form: 
$$\co(\ao)=\sum_i (\ao_i-\io_i)^2 
  + 2 \sum_{(i,j)\in E, i>j} w_{i,j} (\ao_i-\ao_j)^2.$$
We can write this concisely in matrix form, by using
the {\em weighted Laplacian matrix} $L$ of $G$.
$L$ is defined by setting 
$L_{i,i} = \sum_{j\in N(i)} w_{i,j} $ and  $L_{i,j} = - w_{i,j}$. 
We can thus write the social cost as 
$\co(\ao) =\ao^TA\ao + ||\ao-\io||^2$, where the matrix $A=2L$ captures 
the tension on the edges.
The optimal solution is the $\oo$ minimizing $\co(\cdot)$. By taking 
derivatives, we see that the optimal solution satisfies
$(A+I)\oo = \io$. Since the Laplacian of a graph is a
positive semidefinite matrix, it follows that $A+I$ 
is positive definite.
Therefore, $(A+I)\oo = \io$ has a unique solution: $\oo =  (A+I)^{-1}
s$.

In the Nash equilibrium $\no$ each player chooses an opinion which
minimizes her cost; in terms of the derivatives of the cost functions,
this implies that $\co'_i(\no)=0$ for all $i$. 
Thus, to find the players' opinions in the Nash equilibrium we
should solve the following system of equations: $\forall i~
(\no_i-\io_i)  +\sum_{j\in N(i)} w_{i,j}(\no_i-\no_j) = 0 $. Therefore
in the Nash equilibrium each player holds an opinion which is a
weighted average of her internal opinion and the Nash equilibrium
opinions of all her neighbors. This can be succinctly written as
$(L+I)\no = (\frac{1}{2}A+I)\no = \io$.  As before $\frac{1}{2}A+I$ is a
positive definite matrix, and hence the unique Nash equilibrium is $\no =
(\frac{1}{2}A+I)^{-1} \io$.

We now begin our discussion on the price of anarchy (PoA) of the \game ---
the ratio between the cost of the optimal solution and the cost of the Nash equilibrium. 

Our main theorem is that the price of anarchy of the \game is at most
$9/8$. Before proceeding to prove the theorem we present a simple upper
bound of $2$ on the PoA for undirected graphs. To see why this holds, note that
the Nash equilibrium actually minimizes the function
$\ao^T(\frac{1}{2}A)\ao + ||\ao-\io||^2$ (one can verify that this
function's partial derivatives are the system of equations defining the
Nash equilibrium). This allows us to write
the following bound on the PoA:
\begin{eqnarray*}
PoA=\dfrac{\co(\no)}{\co(\oo)} &\leq& \dfrac{2 ( \no^T(\frac{1}{2}A)\no + ||\no-\io||^2)}{\co(\oo)} \\
&\leq& \dfrac{2 ( \oo^T(\frac{1}{2}A)\oo + ||\oo-\io||^2)}{\co(\oo)} \\
&\leq& \dfrac{2\co(\oo)}{\co(\oo)} = 2.
\end{eqnarray*}
We note that 
this bound holds only for the undirected case, 
as in the directed
case the Nash equilibrium does not minimize 
$\ao^T(\frac{1}{2}A)\ao + ||\ao-\io||^2$ anymore.

We now state the main theorem of this section.
\begin{theorem} \label{PoAlongProof}
For any graph $G$ and any internal opinions vector $\io$, the price of anarchy of the \game is at most $9/8$. 
\end{theorem}
\begin{proof}
The crux of the proof is relating the price of anarchy of an instance
to the eigenvalues of its Laplacian. Specifically, we 
characterize the graphs and internal opinion vectors with maximal PoA.
In these worst-case instances at least one eigenvalue of the Laplacian
is exactly $1$, and the vector of internal opinions is a 
linear combination of the eigenvectors associated with the eigenvalues $1$,
plus a possible constant shift for each connected component.
As a first step we consider two matrices $B$ and $C$ that
arise by plugging the Nash equilibrium and optimal solution we
previously computed into the cost function and applying simple
algebraic manipulations:
\begin{align*}
\co(\oo) 
&=  \io^T [ \underbrace{(A+I)^{-1} -I)^2 + (A+I)^{-1} A (A+I)^{-1}}_{B}]\io \\
\co(\no) 
 &= \io^T [ \underbrace{(L+I)^{-1} -I)^2
                        + (L+I)^{-1} A (L+I)^{-1}}_{C}] \io.
\end{align*}

Next, we show that the matrices $A,B,C$ are 
{\em simultaneously diagonalizable}: there exists an
orthogonal matrix $Q$ such that
$A=Q \Lambda^A Q^T$, $B=Q \Lambda^B Q^T$ and $C=Q \Lambda^C Q^T$, where 
for a matrix $M$ the notation
$\Lambda^M$ represents a diagonal matrix with the eigenvalues 
$\lambda_1^M, \ldots, \lambda_n^M$ of $M$ on the diagonal. 

 \begin{lemma} \label{simultaneously_diagonalizable}
$A$,$B$ and $C$ are simultaneously diagonalizable by a matrix $Q$ 
whose columns are eigenvectors of $A$.
\end{lemma}
\begin{proof}
It is a standard fact that any real symmetric matrix $M$ can be diagonalized by an orthogonal matrix $Q$ such that $M =Q\Lambda^M Q^T$. $Q$'s columns are eigenvectors of 
$M$ which are orthogonal to each other and have a norm of one. Thus in order to show that $A$, $B$ and $C$ can be diagonalized with the same matrix $Q$ it is enough to show 
that all three are symmetric and have the same eigenvectors. For this we use the
following basic fact:

\begin{quote}
{\em 
If $\lambda^N$ is an eigenvalue of $N$,  $\lambda^M$ is an eigenvalue of $M$ and $w$ is an eigenvector of both then:
\begin{enumerate}
\item $\frac{1}{\lambda}$ is an eigenvalue of $M^{-1}$ and $w$ is an eigenvector of $M^{-1}$.
\item $\lambda^N+\lambda^M$ is an eigenvalue of $N+M$ and $w$ is an eigenvector of $N+M$.
\item $\lambda^N \cdot \lambda^M$ is an eigenvalue of $NM$ and $w$ is an eigenvector of $NM$.
\end{enumerate}
}
\end{quote}

From this we can show that any eigenvector of $A$ is also an eigenvector of $B$ and $C$. 
Recall that $A$ is a symmetric matrix, thus, it has $n$ orthogonal eigenvectors which 
implies that 
$A$,$B$ and $C$ are all symmetric and share the same basis of eigenvectors. 
Therefore $A$,$B$ and $C$ are simultaneously diagonalizable.
\end{proof}

We can now express the PoA as a function of the eigenvalues of $B$ and $C$. 
By defining $\io'=Q^T\io$ we have:
\begin{eqnarray*}
PoA &=&\dfrac{\co(\no)}{\co(\oo)} =\dfrac{\io^T C \io}{\io^T B \io}=\dfrac{ \io^T Q \Lambda^C Q^T  \io}{ \io^T Q \Lambda^B Q^T  \io} \\
&=& \dfrac{\io'^T \Lambda^C \io'}{\io'^T \Lambda^B \io'} = \dfrac{\sum_{i=1}^n \lambda_i^C {\io'_i}^2}{\sum_{i=1}^n \lambda_i^B {\io'_i}^2} \leq \max_i \dfrac {\lambda_i^C} 
{\lambda_i^B}
\end{eqnarray*}

The final step of the proof consists of expressing $\lambda_i^C$ and $\lambda_i^B$ as functions of the eigenvalues of $A$ (denoted by $\lambda_i$) and 
finding the value for $\lambda_i$ 
maximizing the ratio between $\lambda_i^C$ and $\lambda_i^B$.
 \begin{lemma}
$\max_i \dfrac {\lambda_i^C} {\lambda_i^B}\leq 9/8$.
The bound is tight if and only if
there exists an $i$ such that $\lambda_i=2$. 
\label{lem:eval-ratio}
 \end{lemma}
 \begin{proof} 
Using the basic facts about eigenvalues which were mentioned in the proof of Lemma \ref{simultaneously_diagonalizable}, we get:
\begin{align*} 
  \lambda^B_i &= \left( 1-\frac{1}{\lambda_i+1} \right)^2+
  \frac{1}{\lambda_i+1} \lambda_i \frac{1}{\lambda_i+1} \\
&= \frac{\lambda_i^2}{(\lambda_i+1)^2}+\frac{\lambda_i}{(\lambda_i+1)^2} 
= \frac{\lambda_i^2+\lambda_i}{(\lambda_i+1)^2} 
= \frac{\lambda_i}{(\lambda_i+1)}
\\
  \lambda^C_i &
= \left(1-\frac{1}{\frac{1}{2} \lambda_i+1}\right)^2+
  \frac{1}{\frac{1}{2}\lambda_i+1} \lambda_i \frac{1}{\frac{1}{2}\lambda_i+1} \\
&= \frac{\lambda_i^2}{(\lambda_i+2)^2}+\frac{4\lambda_i}{(\lambda_i+2)^2} 
= \frac{\lambda_i^2+4\lambda_i}{(\lambda_i+2)^2}.
\end{align*}

We can now write $\lambda_i^C/\lambda_i^B = \phi(\lambda_i)$, 
where $\phi$ is a simple rational function:
\begin{align*} 
  \phi(\lambda) 
  &= \frac{(\lambda^2+4\lambda)/(\lambda+2)^2}
         {\lambda/(\lambda+1)}
  = \frac{(\lambda^2+4\lambda)(\lambda+1)}
         {(\lambda+2)^2 \lambda} \\
  &= \frac{(\lambda+4)(\lambda+1)}
         {(\lambda+2)^2} 
  = \frac{\lambda^2+5\lambda+4}{\lambda^2+4\lambda+4}.
\end{align*} 

By taking the derivative of $\phi$, we find that $\phi$ is maximized over all
$\lambda \geq 0$ at $\lambda = 2$ and $\phi(2) = 9/8$.

The eigenvalues $\lambda_i$ are all non-negative, so 
it is always true that $\max_i \phi(\lambda_i) \leq 9/8$.
If $2$ is an eigenvalue of $A$
(and hence $1$ is an eigenvalue 
of the Laplacian) 
then there exists an internal opinions vector $\io$ for which the 
PoA is $9/8$. 

What is the internal opinions vector maximizing the PoA? To find it assume that the $i$th eigenvalue of the Laplacian equals $1$.
To get a PoA of $9/8$ we should choose $s'_i=1$ and $\forall j\neq i~s'_j=0$ to hit only $\lambda_i$. By definition $\io'=Q^T\io$, 
and hence $\io=(Q^T)^{-1}\io'$. 
Because $Q$ is orthogonal, $Q^T = Q^{-1}$; thus, $\io = Q\io' = v_i$,
where $v_i$ is the eigenvector associated with $\lambda_i$.
In fact, any linear combination of the 
eigenvectors associated with eigenvalues $0$ and $1$ where at least one of 
the eigenvectors of $1$ has a nonzero coefficient will obtain the 
maximal PoA.
\end{proof}

With Lemma \ref{lem:eval-ratio}, we have completed the proof of
Theorem \ref{PoAlongProof}.
\end{proof}

\begin{corollary} \label{scaling}
We can scale the weights of any graph to make its PoA be $9/8$. 
If $\alpha$ is the scaling factor for the weights,
then the eigenvalues of the scaled
$A$ matrix are $\alpha \lambda_i$. Therefore by choosing $\alpha =
\frac{2}{\lambda_i}$ for any eigenvalue other than $0$ we get that
there exists an internal opinions vector for which the PoA is $9/8$.
\end{corollary}

\subsection{Arbitrary Node Weights and Players with Fixed Opinions } 
\label{subsec:fixed}

Our first extension is a model in which different people put different
weights on their internal opinion. In this extension, each
node in the graph has a strictly positive weight $w_i$ and the cost
function is:
$\displaystyle{\co(\ao)=\sum_i [ w_i(\ao_i-\io_i)^2+\sum_{j \in N(i)} w_{i,j}(\ao_i-\ao_j)^2 }]$.
The bound of $9/8$ on the PoA holds even in this model. To see this,
let $w$ be the vector of node weights and $\diag(w)$ be a diagonal
matrix with the values of $w$ on the diagonal. In terms of the scaled
variables
$\hat {\ao} = \diag(\sqrt{w}\,) \ao$, 
$\hat {\io} = \diag(\sqrt{w}\,) \io$
and the scaled matrix
$\hat {A} = \diag(\sqrt{w}\,)^{-1} A  \diag(\sqrt{w}\,)^{-1}$,
the cost takes the same form as before:
$\co(\hat \ao)=\|\hat \ao-\hat \io\|^2 +\hat {\ao}^T \hat {A} \hat {\ao} $. 
We have therefore proved:
\begin{claim}
The PoA of the game with arbitrary strictly positive node weights is bounded by $9/8$.
\label{nodeWeightsPositive}
\end{claim}

Next we show how to handle the case in which a subset of the players may
have node weights of $0$, which can equivalently be viewed as
a set of players who have no internal opinion at all.
We analyze this by first considering the case in which all non-zero
node weights are the same; for this case we prove: 
\begin{lemma} \label{lem:nodeWeightsZero}
If every player has either weight $1$ or $0$ on her internal opinion, then the PoA is bounded by $9/8$.
\end{lemma}
\begin{proof}
Let $Z$ be the set of players who do not have an internal opinion. We define the following diagonal matrix $R$: $R_{i,i}=0$ for  $i \in Z$ and $R_{j,j}=1$ for $j \notin Z$. We 
assume 
without loss of generality that $Z \neq V$ since otherwise the PoA is 1. We can also assume without loss of generality that in the instance which maximizes the PoA each $i
\in Z$ has an internal opinion of $0$. Therefore we can express the social cost as $\co(\ao)=|| \ao-\io||^2 + \ao^T (A+R-I) \ao$. Since the cost associated with all $i \notin Z$ 
remains the same while for $i\in Z$ the cost of $|| \ao_i-\io_i||^2 = \ao_i^2$ is countered by the $-\ao_i^2$ from the $i^{th}$ row of $\ao^T(R-I)\ao$. Similar to before we have that 
the optimal solution is 
$\oo=(A+R)^{-1}\io$ and the Nash equilibrium is $\no=(\frac{1}{2}A+R)^{-1}\io$. Since any real vector is an eigenvector of $R$ we have that $(A+R-I)$, $(A+R)$ and $(\frac{1}{2}A+R)$ are 
simultaneously diagonalizable and therefore the same steps 
we took to prove Theorem \ref{PoAlongProof} lead us to get a bound 
of $9/8$ for the PoA.
\end{proof}

By applying the change of variables from Claim \ref{nodeWeightsPositive} we can also handle non-zero arbitrary weights.

In the second model we present, some nodes have {\em fixed opinions}
and others do not have an internal opinion at all.
We partition the nodes into two sets $A$ and $B$.  Nodes in $B$ are completely
fixed in their opinion and are non-strategic, while nodes in $A$
have no internal opinion -- they simply want to choose an opinion
that minimizes their disagreement with their neighbors
(which may include a mix of nodes in $A$ and $B$). 
We can think of nodes in $A$ as people forming their opinion and
of nodes in $B$ as news sources
with a specific \emph{fixed} orientation. 
We denote the fixed opinion of a node $j \in B$ by $\io_j$.
The social cost for this model is:
$$\co(\ao)= \sum_{\substack{(i,j) \in E; \\  i \in A; j \in B}} (\ao_i-\io_j)^2+2\sum_{\substack{(i,j) \in E; \\ i,j \in A; i>j}}  (\ao_i-\ao_j)^2.$$

Note that this clearly generalizes the original model, since we can
construct a distinct node in $B$ to represent each internal opinion.
Next, we perform the reduction in the opposite direction, reducing
this model to the basic model. 
To do this, we assign each node an internal opinion equal to the weighted
average of the opinions of her fixed neighbors, and a weight equal to the
sum of her fixed neighbors' weights.
We then show that the PoA of the fixed opinion model is bounded by the PoA of the basic model and thus get:
\begin{proposition}
The PoA of the fixed opinion model is at most $9/8$.
\label{prop:fixedOpinions}
\end{proposition}
\begin{proof}
We reduce an instance of the fixed opinion game to an instance of the opinion game with arbitrary node weights as follows:  
We define the internal opinion of every player $i \in A$ that has at least one neighbor in $B$ to be a 
weighted average of the opinions $i$'s neighbors in $B$: 
$\io_i=\dfrac {\sum_{j\in N_B(i)} w_{i,j} \io_j } {\sum_{j\in N_B(i)} w_{i,j}}$, where $N_B(i)$ is the set of $i$'s neighbors in $B$. We also define node $i$'s weight to be $w_i=\sum_{j\in N_B(i)} 
w_{i,j}$. 
For a player $i \in A$ who does not have any neighbors in $B$ we simply define $w_i=0$.
We use $G$ 
to denote the initial instance, and $G'$ to denote the instance produced by
the reduction.
Let $\no$ be the Nash equilibrium in $G$; then $\no$ is also the
Nash equilibrium in $G'$. To see this, recall that in a Nash
equilibrium each player's opinion is the weighed average of the opinions
of all neighbors. Thus, 
$$\no_i=\dfrac{\sum_{j \in N_B(i)}
w_{i,j}s_j  +\sum_{j \in  N_A(i)} w_{i,j}y_j  }{\sum_{j \in  N_B(i)}
w_{i,j}+\sum_{j \in  N_A(i)} w_{i,j}}=\dfrac{  (\sum_{j \in  N_B(i)} w_{i,j}) \frac {\sum_{j\in N_B(i)}
w_{i,j} \io_j } {\sum_{j\in N_B(i)} w_{i,j}} +\sum_{j \in  N_A(i)}
w_{i,j}y_j}{\sum_{j \in  N_B(i)}
w_{i,j}+\sum_{j \in  N_A(i)} w_{i,j}}.$$

In Claim \ref{clm-bounded-cost} below we show that $\co_G(\ao)=\co_{G'}(\ao)+c$ for a positive constant $c$, hence, the optimal solution for $G$ and $G'$ is the same.
Let $\oo$ be this optimal solution and let $\no$ be $G$'s and $G'$'s Nash equilibrium. By deriving the following bound we conclude the proof:
$$ PoA(G)=\dfrac{\co_G(\no)}{\co_G(\oo)} \leq \dfrac {\co_{G'}(\no)+c} {\co_{G'}(\oo)+c} \leq   \dfrac {\co_{G'}(\no)} {\co_{G'}(\oo)}  \leq \dfrac {9}{8}.$$

\begin{claim} \label{clm-bounded-cost}
$\co_G(\ao)=\co_{G'}(\ao)+c$ where $c$ is a positive constant.
\end{claim}
\begin{proof}
We show that $\co_G(\ao) \geq \co_{G'}(\ao)$ and $\co_G(\ao)-\co_{G'}(\ao)$ is constant. Observe that the only terms 
where the two costs differ are ones associated with the cost of the 
fixed opinions in $G$ and the internal opinions in $G'$. 
Thus, it is enough to show that for every player $i$:
$$\sum_{j\in N_B(i)} w_{i,j}(\io_j-\ao_i)^2  \geq  \left(\sum_{j\in N_B(i)} w_{i,j}\right) \cdot  \left(\dfrac {\sum_{j\in N_B(i)} w_{i,j} \io_j} {\sum_{j\in N_B(i)} w_{i,j}} -\ao_i\right)^2 .$$
By arranging the terms we get that the terms involving $z_i$'s cancel out, hence what left to show is that:  
$\sum_{j\in N_B(i)} w_{i,j}\io_j^2  \geq  \dfrac {(\sum_{j\in N_B(i)} w_{i,j} \io_j)^2 } {\sum_{j\in N_B(i)} w_{i,j}} $.
The claim follows from the following computation:
\begin{eqnarray*}
\left(\sum_{j\in N_B(i)} w_{i,j} \io_j\right)^2 &=&  \sum_{j\in N_B(i)}  w_{i,j}^2 \io_j^2 + \sum_{j,k\in N_B(i), j\geq k}  2w_{i,j}w_{i,k} \io_j \io_k \\
&\leq&   \sum_{j\in N_B(i)}  w_{i,j}^2 \io_j^2 + \sum_{j,k\in N_B(i), j\geq k}  w_{i,j}w_{i,k} (\io_j^2+ \io_k^2) \\
&=& \left(\sum_{j\in N_B(i)} w_{i,j}\right) \left(\sum_{j\in N_B(i)} w_{i,j}\io_j^2\right).
\end{eqnarray*}
\end{proof}
\end{proof}
\section{Directed Graphs} \label{sec:directed}

We begin our discussion of directed graphs with an example showing
that the price of anarchy can be unbounded even for graphs with 
bounded degrees.
Our main result in this section is that we can nevertheless
develop spectral methods extending those in Section~\ref{sec:undirected}
to find internal opinions that maximize the PoA for a given graph.
Using this approach, we identify classes of directed graphs with
good PoA bounds.

In the introduction we have seen that the PoA of an in-directed star
can be unbounded. As a first question, we ask whether this is
solely a consequence of the unbounded maximum in-degree of this graph,
or whether it is
possible to have an unbounded PoA for a graph with bounded degrees.
Our next example shows that one can obtain a large PoA even when all
degrees are bounded: we show that the PoA of a 
bounded degree tree can be $\Theta(n^c)$, where $c\leq1$ 
is a constant depending on the in-degrees of the nodes in the tree. 
\begin{example}
Let $G$ be a $2^k$-ary tree of depth $\log_{2^k} n$ in which the
internal opinion of the root is $1$ and the internal opinion of every other node
is $0$. All edges are directed toward the root. In the Nash
equilibrium all nodes at layer $i$ hold the same opinion, which is
$2^{-i}$. (The root is defined to be at layer $0$.) The cost of a
node at layer $i$ is $2\cdot 2^{-2i}$. Since there are $2^{ik}$ nodes
at layer $i$,
the total social cost at Nash equilibrium is
$\displaystyle{\sum_{i=1}^{\log_{2^k} n} 2^{ik}2^{1-2i} =
2\sum_{i=1}^{\log_{2^k} n} 2^{(k-2)i}.}$
For $k>2$ this cost is 
$\displaystyle{2^{k-1} \dfrac{(2^{k-2})^{\log_{2^k} n}-1}{2^{k-2}-1} =2^{k-1} \dfrac{n^\frac{k-2}{k}-1}{2^{k-2}-1}}.$
The cost of the optimal solution is at most $1$; in fact it is very
close to $1$, since in order to reduce the cost the root should hold
an opinion of $\epsilon$ very close to $0$, which makes the root's
cost approximately $1$. Therefore the PoA is $\Theta(n^\frac{k-2}{k})$. 
It is instructive to consider the PoA for extreme values of $k$. 
For $k=2$, the PoA is $\Theta(\log n)$, while for $k = \log n$
we recover the in-directed star from the introduction where the PoA 
is $\Theta(n)$.
For intermediate values of $k$, the PoA is $\Theta(n^c)$.
For example, for $k=3$ we get that the PoA
is $\Theta(n^\frac{1}{3})$.  
\end{example}

For directed graphs we do not consider the generalization to arbitrary node
weights (along the lines of Section \ref{subsec:fixed}), noting instead
that introducing 
node weights to directed graphs can have a severe effect on the PoA. 
That is, even in graphs containing only two nodes, introducing
arbitrary node weights can make the PoA unbounded. 
For example, consider a graph with two
nodes $i$ and $j$. Node $i$ has an internal opinion of $0$ and a node weight
of $1$, while node $j$ has an internal opinion of $1$ and a node weight of
$\epsilon$. There is a directed edge $(i,j)$ with weight $1$. 
There cost of the Nash equilibrium is $1/2$, but the social cost of the
optimal solution is smaller than $\epsilon$. 
To avoid this pathology, from now on we restrict our attention to uniform node weights.

\subsection{The Price of Anarchy in a General Graph}
For directed graphs the cost of the optimal solution and the
cost of the Nash equilibrium are respectively $c(\oo) = \io^T B \io$
and $c(\no) = \io^T C \io$, as before. But now, $C$ has a slightly more
complicated form since $L$ is no longer a symmetric matrix. 
Recall that matrix $A$ is used in the social cost function 
to capture the cost associated with the edges
of the graph (disagreement between neighbors). 
We define it for directed graph by setting
$A_{i,j} = -w_{i,j}-w_{j,i}$ for $i \neq j$ and
$A_{i,i} = \sum_{j \in N(i)} w_{i,j} + \sum_{\{j|i \in N(j)\}} w_{j,i}$.
The matrix $A$ is the weighted Laplacian for an undirected graph
where the weight on the undirected edge $(i,j)$ is the sum of the
weights in the directed graph for edges $(i,j)$ and $(j,i)$.
We then define
$C = \left( (L+I)^{-1} -I \right)^T \left( (L+I)^{-1} -I \right) + (L+I)^{-T} A (L+I)^{-1}.$ 
The price of anarchy, therefore, is $\dfrac{\io^T C \io}{\io^T B \io}$ as before. 
The primary distinction between the price of anarchy in the directed
and undirected cases is that in the undirected case, $B$ and $C$ are
both rational functions of $A$. In the directed case, no such simple
relation exists between $B$ and $C$, so that we cannot easily bound
the generalized eigenvalues for the pair (and hence the price of
anarchy) for arbitrary graphs. However, given a directed graph 
our main theorem shows that we can always find the vector of internal 
opinions $\io$ yielding the maximum PoA:
\begin{theorem}
Given a graph $G$ it is possible to find the internal opinions vector $\io$ yielding the maximum PoA up to a precision of $\epsilon$ in polynomial time.
\end{theorem}
\begin{proof}
The total social cost is invariant under constant shifts in opinion.
Therefore, without loss of generality, we restrict our attention to
the space of opinion vectors with mean zero. Let us define a matrix
$P \in \mathbb{R}^{n \times (n-1)}$ to have $P_{j,j} = 1$, $P_{j+1,j} = -1$, 
and $P_{i,j} = 0$ otherwise. The columns of $P$ are a basis for the space of
vectors with mean zero; that is, we can write any such vector as 
$\io = P \hat{\io}$ for some $\hat{\io}$. We also define matrices
$\bar{B} = P^T B P$ and $\bar{C} = P^T C P$, which are positive definite
if the symmetrized graph is connected. The price of anarchy is then
given by the generalized Rayleigh quotient
$\rho_{\bar{C},\bar{B}}(\hat{\io}) =
 (\hat{\io}^T \bar{C} \hat{\io})/(\hat{\io}^T \bar{B} \hat{\io})$.
Stationary points of $\rho_{\bar{C},\bar{B}}$ satisfy the generalized
eigenvalue equation 
$(C-\rho_{\bar{C},\bar{B}}(\hat{\io}) \bar{B}) \hat{\io} = 0$. In
particular, the price of anarchy is the largest generalized
eigenvalue, and the associated eigenvector $\hat{\io}_*$ corresponds
to the maximizing choice of internal opinions.

The solution of generalized eigenvalue problems is a standard technique
in numerical linear algebra, and there are good algorithms that run in
polynomial time; see~\cite[\S8.7]{GVL}. In particular, because $\bar{B}$
is symmetric and positive definite, we can use the Cholesky factorization
$\bar{B} = R^T R$ to reduce the problem to the standard eigenvalue problem
$(R^{-T} \bar{C} R^{-1} - \lambda I) (R \hat{\io}) = 0$.
\end{proof}

\subsection{Upper Bounds for Classes of Graphs} 
Our goal in this section is rather simple: we would like to find
families of graphs for which we can bound the price of anarchy.  The
main tool we use is bounding the cost of the Nash equilibrium by a
function of a simple structure. By using a function that has a similar
structure to the social cost function we are able to frame the bound
as a generalized eigenvalue problem that can be solved using
techniques similar to the ones that were used in proving Theorem
\ref{PoAlongProof}.

\begin{proposition}  \label{prop:directed_bound}
Let $\mathcal{G}$ be a graph family for which there exists a $\beta$ such that for any $G \in \mathcal{G}$ and any internal opinions vector $s$, we have
$\co(\no) \leq \min_\ao (\beta (\ao^T A \ao) + ||\ao-\io||^2)$. Then, $\forall G \in \mathcal{G}$ and opinion vectors $s$, $\poa(G) \leq
\frac{ \beta +\beta \lambda_2}{1+\beta \lambda_2}$, where $\lambda_2$
is the second smallest eigenvalue of $A$. 
\end{proposition}
%
\begin{proof}
Let $\tilde \oo$ be the vector minimizing $\beta (\ao^T A \ao) + ||\ao-\io||^2$. We can derive the following bound on the price of
anarchy: 
\begin{align*}
PoA(G) &=\dfrac{\co(\no)}{\co(\oo)} \leq \dfrac{\beta(\tilde \oo^T A \tilde \oo) + ||\tilde \oo-\io||^2}{(\oo^T A \oo) +
||\oo-\io||^2} = \dfrac{s^T C s}{s^T B s},
\end{align*}
where $C$ and $B$ are defined similarly to the matrices in Theorem \ref{PoAlongProof} and are simultaneously diagonalizable. If $\lambda_i$ is an eigenvalue of $A$ then $
\lambda^B_i=\frac{\lambda_i}{1+\lambda_i} $ and $\lambda^C_i=\frac{\beta\lambda_i}{1+\beta\lambda_i} $. As before, the maximum PoA is achieved when $\lambda^C_i /\lambda^B_i = \frac{\beta\lambda_i}
{1+\beta\lambda_i} /   \frac{\lambda_i}{1+\lambda_i}  = \frac{\beta \lambda_i +\beta}{\beta\lambda_i +1}$ is maximized. The maximum here is taken over all eigenvalues different than $0$ as we know that the PoA for the internal opinions vector associated with eigenvalue $0$ (which is a constant vector) is $1$. Therefore, the maximizing eigenvalue is $\lambda_2$.

\end{proof}

An immediate corollary is that if there exists a $\beta$ as in Proposition \ref{prop:directed_bound} then the PoA is bounded by this $\beta$.

We say that a bounded degree {\em asymmetric} expander
is an unweighted directed graph that does not contain any pair of 
oppositely oriented edges $(i,j)$ and $(j,i)$, and whose
symmetrized graph
has maximum degree $\Delta$ and edge expansion $\alpha$. We show:
\begin{claim} \label{clm:expanders_gen}
For a bounded degree asymmetric expander the PoA is bounded by $O(\Delta^2 / \alpha^2) $. 
\end{claim}
\begin{proof}
For an asymmetric graph, the matrix $A$ is simply the Laplacian of the 
underlying graph; 
this is why we require in the claim that the graph is asymmetric.

If $\Delta$ is the maximum degree, then we have $\lambda_2 \leq \lambda_n\leq \Delta$. 
We also have that $\lambda_2 \geq \alpha^2 / 2\Delta$
\cite{chung-spectral-graph-theory}.
We can now use this to bound the PoA in terms of the graph's expansion as follows:
\begin{align*}
\dfrac{\beta + \beta \lambda_2}{1+\beta \lambda_2} &\leq \dfrac {\beta + \beta \lambda_2}{\beta \lambda_2} \leq \dfrac{1+\lambda_2}{\lambda_2} 
\leq \dfrac{2\Delta(1+\Delta)}{\alpha^2} = O(\Delta^2 / \alpha^2).
\end{align*}
\end{proof}

The next natural question is for which graph families such a $\beta$
exists. Intuitively, such a $\beta$ exists whenever the cost of the
Nash equilibrium is smaller than the cost of the best consensus --- that
is, the optimal solution restricted to opinion vectors in which all
players hold the same opinion (constant vectors). This is true since
the function $ \beta (\ao^T A \ao) +||\ao-\io||^2$ is the social cost
function of a network in which the weights of all edges have been
multiplied by $\beta$.
However using this intuition for finding graph families for which
$\beta$ exists is difficult and furthermore does not help in computing
the value of $\beta$ (or a bound on it). Hence, we take a different
approach. In Lemma \ref{lem:schur}, we introduce an intermediate
function $g(\cdot)$ with the special property that its minimum value
is the same as the cost of the Nash equilibrium. By showing that there
exists a $\beta$ such that $g(\ao) \leq \beta \ao^T A \ao
+||\ao-\io||^2$ we are able to present bounds for Eulerian
bounded-degree graphs and additional bounds for Eulerian
bounded-degree asymmetric expanders.
As a first step, 
we use Schur complements to prove the 
following: 


\begin{lemma} \label{lem:schur}
Let $g(\ao)=\ao^T M \ao + ||\ao-\io||^2$ with $M = (I-C)^{-1}-I$. 
If $(I-C)$ is nonsingular then for the Nash equilibrium $\no$, we have $\min_{\ao} g(\ao) = \co(\no) $.
\end{lemma}

\begin{proof}
The social cost is a quadratic function
of the expressed opinion vector and the internal opinion vector:
\[
  c(\ao) = \ao^T A \ao + \|\ao-\io\|^2
         = \begin{bmatrix} \ao \\ \io \end{bmatrix}^T
           \begin{bmatrix} A+I & -I \\ -I & I \end{bmatrix}
           \begin{bmatrix} \ao \\ \io \end{bmatrix}.
\]
To compute the socially optimal vector, we minimize this quadratic
form in $\ao$ and $\io$ subject to constraints on $s$.
This yields $c(\oo) = \io^T B \io$, where the matrix
\[
  B = ((A+I)^{-1} -I)^2 + (A+I)^{-1} A (A+I)^{-1} = I - (A+I)^{-1}
\]
is a {\em Schur complement} in the larger system\footnote{Recall that the Schur complement of the block A of the matrix $ \begin{bmatrix} A & B\\ C & D \end{bmatrix}$ is $D-CA^{-1}B$. }.
  Schur complements
typically arise in partial elimination of variables from linear systems.
In this case, we have eliminated the $z$ variables in the stationary
equations for a critical point in the extended quadratic form.

Now consider the Nash equilibrium. As we assume that $(I-C)$ is invertible, we can define
\[
  M = (I-C)^{-1}-I.
\]
The matrix $M$ is symmetric and positive semidefinite, with a null
space consisting of 
the constant vectors.
That is, we can see
$M$ as the Laplacian of a new graph.
By design, $C = I-(M+I)^{-1}$, so we can mimic the construction above
to express $C$ as a Schur complement in a larger system. 
Thus, the social cost of the Nash equilibrium can be written
\[
  c(\no) = \min_{\ao}
           \begin{bmatrix} \ao \\ \io \end{bmatrix}^T
           \begin{bmatrix} M+I & -I \\ -I & I \end{bmatrix}
           \begin{bmatrix} \ao \\ \io \end{bmatrix},
\]
which is the optimal social cost in the new network.
\end{proof}

We then complement the lemma by showing that for Eulerian graphs $(I-C)$ is nonsingular and furthermore the matrix $M$ has a nice structure:
\begin{claim} \label{clm:m-is-nice}
For Eulerian graphs $M= A+LL^T$.
\end{claim}
\begin{proof}
We denote $\tilde L = L+I$ and $\tilde A = A+I$ then:
\begin{eqnarray*}
I-C &=& I - (\tilde L^{-1} -I)^T (\tilde L^{-1} -I) -\tilde L^{-T} (\tilde A - I) \tilde L^{-1} \\
&=& \tilde L ^{-1} + \tilde L ^{-T} - \tilde L^{-T} \tilde A \tilde L^{-1}.
 \end{eqnarray*}
We use the fact that for Eulerian graphs $A=L+L^T$ which implies that $\tilde A= \tilde L + \tilde L^T - I$ to simplify $I-C$:
\begin{eqnarray*}
I-C &=&  \tilde L ^{-1} + \tilde L ^{-T} - \tilde L^{-T} (\tilde L + \tilde L^T - I) \tilde L^{-1}\\
&=& \tilde L^{-T}   \tilde L^{-1}.
 \end{eqnarray*}
 We have that $M= (L+I) (L+I)^T-I = A+LL^T$. Let us understand what
the matrix $LL^T$ looks like. On the diagonal we have $ [LL^T]_{i,i} =  d_i^2 + d_i$ where $d_i$ is the degree of node $i$ and off the diagonal 
$ [LL^T]_{i,j} =  d_i L_{j,i} +d_j L_{i,j} + \sum_{k\neq i, j} (L_{i,k} L_{j,k}) =  d_i L_{j,i} +d_j L_{i,j} +|N(i) \cap N(j)|$.
\end{proof}

Recall that $\Delta$ is the he maximum degree of an Eulerian graph. We are now ready to prove the following proposition:

\begin{proposition}
For Eulerian graphs $\co(\no) \leq (\Delta+1) (\ao^T A \ao) +||\ao-\io||^2$. 
\end{proposition}
\begin{proof}
By Lemma \ref{lem:schur} and Claim \ref{clm:m-is-nice} we have that for Eulerian graphs
 $\co(\no)=\min_z g(\ao)=\min_z \ao^T (A+LL^T) \ao + ||\ao-\io||^2$. 
 What remains to show is that for $\beta = \Delta+1$ it holds that
 $g(\ao) \leq \beta \ao^T A \ao + ||\ao-\io||^2$.
  After some rearranging this boils down to showing that the following holds
$\ao^T LL^T \ao \leq (\beta -1) \ao^T A \ao$

Note that $A$ is the Laplacian for a symmetrized version of the graph;
assuming this graph is connected (since otherwise we can work separately
in each component), this means $A$ has one zero
eigenvalue corresponding to the constant vectors, and is positive
definite on the space orthogonal to the constant vector.  Similarly,
$LL^T$ has a zero eigenvalue corresponding to the constant vectors, and is
at least positive semi-definite on the space orthogonal to the constant
vectors.  Since $A$ is positive definite on the space of nonconstant
vectors, the smallest possible $\beta$ can be computed via the solution of
a generalized eigenvalue problem
\[
   \beta = 1 + \max_{z \neq \alpha e} \frac{z^T LL^T z}{z^T A z}.
\]

In the case of an unweighted graph, one get a bound
via norm inequalities.  Using the fact that the graph is Eulerian,
$L^T$ is also a graph Laplacian, and we can write
\[
  \left( L^T z \right)_i = \sum_{j=1}^n w_{j,i} (z_i-z_j),
\]
so
\[
  z^T L L^T z = \sum_{i=1}^n \left( \sum_{j=1}^n w_{j,i} (z_i-z_j) \right)^2.
\]
Similarly, we expand the quadratic form $z^T A z$ into
\[
  z^T A z = \sum_{i < j} (w_{i,j} + w_{j,i}) (z_i-z_j)^2
          = \sum_{i=1}^n \left( \sum_{j=1}^n w_{j,i} (z_i-z_j)^2 \right).
\]
Now, recall that in general $ \left( \sum_{j=1}^{d} x_j \right)^2 \leq d \sum_{j=1}^d x_j^2,$
which means that in the unweighted case
$
\left( \sum_{j=1}^n w_{j,i} (z_i-z_j) \right)^2 \leq d_i
\left( \sum_{j=1}^n w_{j,i} (z_i-z_j)^2 \right).
$
where $d_i = \sum_{j} w_{j,i}$ is the in-degree or out-degree.
Therefore,
\[
  \frac{z^T L L^T z}{z^T A z} \leq
  \frac{\sum_{i=1}^n d_i \sum_{j=1}^n w_{j,i} (z_i-z_j)^2}
       {\sum_{i=1}^n \sum_{j=1}^n w_{j,i} (z_i-z_j)^2} \leq
  \max_i d_i = \Delta.
\]
So for a general Eulerian graph, $\beta \leq 1+\Delta$.

\end{proof}


We observe that for a cycle the bound of $2$ on the price of anarchy is actually tight:
\begin{observation}
The PoA of a directed cycle is bounded by $2$ and approaches $2$ as the size of the cycle grows.
\end{observation}
\begin{proof}
For a cycle it is the case that $A=LL^T$; therefore $g(\ao)=2(\ao^T A \ao) +
||\ao-\io||^2$, and hence the bound assumed in Proposition
\ref{prop:directed_bound} is actually a tight bound. In order to show that
the PoA indeed approaches $2$ we need to show that $\lambda_2$
approaches $0$ as the size of the cycle grows. The fact that $A$ is
the Laplacian of an undirected cycle comes to our aid and provide us
an exact formula for $\lambda_2$:  $\lambda_2 = 2(1-\cos
(\frac{2\pi}{n}))$ (where $n$ is the size of the cycle), and this
concludes the proof.  
\end{proof}


For general Eulerian graphs we leave open the question of whether the
bound of $\Delta +1$ is a tight bound or not. 
Indeed, it is an intriguing open question whether there
exists a Eulerian graph with PoA greater than $2$.

\section{Adding Edges to the Graph } \label{sec:adding-edges}
The next thing we consider is the following class of problems:
Given an unweighted graph $G$ and a vector of internal opinions $\io$, 
find edges $E'$ to add to $G$ so as to minimize the social cost of the Nash equilibrium.
We begin with a 
general bound linking the possible 
improvement from adding edges to the price of anarchy. Let $G$ be a graph (either undirected or directed). Denote by $\co_G(\cdot)$ the cost function and by $\no$ and $\oo$ the 
Nash equilibrium and optimal solution respectively. Let $G'$ be the graph constructed by adding edges to $G$. Then: $\displaystyle{\dfrac{c_G(\no)}{c_{G'}(\no')} \leq \dfrac{c_G
(\no)}{c_{G'}(\oo')} \leq \dfrac {c_G(\no)}{c_G(\oo)} =\poa(G)}$. 
To see why this is the case, 
we first note that $c_{G'}(\oo') \leq c_{G'}(\no')$ since the cost of the Nash equilibrium cannot be smaller than the optimal solution. Second, $c_G(\oo) \leq c_{G'}(\oo')$ simply 
because $c_{G'}(\cdot)$ contains more terms than $c_{G}(\cdot)$. Therefore we have proved the following proposition:
\begin{proposition}
Adding edges to a graph $G$ can improve the cost of the Nash equilibrium by a multiplicative factor of at most the PoA of $G$.
\end{proposition}

We study three variants on the problem, discussed in the introduction. 
In all variants, we seek the ``best'' edges to add in order to minimize
the social cost of the Nash equilibrium. 
The variants differ mainly in the types of edges we may add.

\xhdr{Adding edges from a specific node}
First, we consider the case in which we can only add edges from a specific node $w$. Here
we imagine that node $w$ is a media source that therefore 
does not have any cost for holding 
an opinion, and so we will use a cost function
that ignores the cost associated with it 
when computing the social cost.
Hence, our goal is to find a set of nodes $F$ such that adding edges from node $w$ to all the nodes in $F$
minimizes the cost of the Nash equilibrium while ignoring the cost exhibited by $w$.
By reducing the subset sum problem to this problem we show that:
\begin{proposition} \label{prop-from-is-np}
Finding the best set of edges to add from a specific node $w$ is NP-hard.
\end{proposition}
\begin{proof}
Denote by $G+F$ the graph constructed by adding to $G$ edges from $w$ to all nodes in $F$. 
Our goal is to find a set $F$ minimizing $\coa_{G+F}(\no)$, where $\no$
is a Nash equilibrium in the graph $G+F$ and $\coa$ denotes the
total cost of all nodes in $\no$ except for node $w$.
We show that finding this
set is NP-hard by reducing the subset sum problem to this problem.
Recall that in the subset sum problem we are given a set of positive integers $a_1,\dots,a_n$ and a number $t$. We would like to know if there exists any subset $S$ such that $
\sum_{j \in S} a_j = t$.
Given an instance of the subset problem, we reduce it to the following instance of the opinion game. The
instance conatins an in-directed star with $n$ peripheral nodes that
have an internal opinion of $0$ and a center node $w$ which has an
internal opinion of $1$ and $n$ isolated nodes
that have internal opinions of $-\frac{a_i}{t}$.

\begin{lemma}
For the graph $G$ and the vector of internal opinions $s$ defined above, there exists a set $F$ such that $\coa_{G+F}(\no) = 0$ if and only if the answer to the subset problem is 
yes.
\end{lemma}
\begin{proof}
As seen in the introduction, in the Nash equilibrium each one of the peripheral nodes holds an opinion of $\frac{1}{2}x_w$. Node $w$ hold an opinion of $\no_w=\dfrac{1+\sum_{j
\in F} \io_j }{1+|F|}$. Therefore the cost of the Nash equilibrium in $G+F$ is:
$$\coa_{G+F}(\no) = n \left( (\dfrac{1}{2} \no_w-0)^2 + (\no_w-\dfrac{1}{2}\no_w)^2 \right)  = 2n\left(\dfrac{1+\sum_{j\in F} \io_j }{2(1+|F|)} \right)^2.$$ 
Clearly the cost is nonnegative as it is a sum of quadratic terms;
moreover it equals $0$ if and only if $\sum_{j\in F} \io_j = -1$.
Defining $F'=\{j \in F | s_j<0\}$, we have $\sum_{j\in F'} \io_j = -1$. By
the reduction we have that $\sum_{j\in F'} -\frac{a_j}{t} = -1$; if we
multiply by $-t$ we get that $\sum_{j\in F'} a_j = t$ implying that
there exists a solution to the subset sum problem.
\end{proof}
\end{proof}

\xhdr{Adding edges to a specific node}
Next, we consider the case in which we can only add edges to a
specific node. We can imagine again that node $w$ is a media source;
in this case, however, our goal is to find the best set of people to
expose to this media source. By reducing the minimum vertex cover
problem to this problem we show that:

\begin{proposition} \label{prop-to-is-np}
Finding the best set of edges to add to a specific node $w$ is NP-hard.
\end{proposition}
\begin{proof}
%
Given an instance of the minimum vertex cover problem, consisting
of an undirected graph $G'=(V',E')$, we
construct an instance of the opinions game as follows:
\begin{itemize}
\item For each edge $(i,j)\in E'$ we create a vertex $v_{i,j}$ with internal opinion $1$.
\item For every $v_{i,j}$ we create an in-directed star with $24$ peripheral nodes that have an internal opinion of $0$.
We later refer to node $v_{i,j}$ and all the nodes directed to it as {\em $v_{i,j}$'s star}.
\item For each vertex $i \in V'$ we create a vertex $u_i$ with internal opinion $1$.
\item For each edge $(i,j) \in E'$ we create directed edges $(v_{i,j},u_i)$ and $(v_{i,j},u_j)$.
\item We create an isolated node $w$ with internal opinion $-3$.
\end{itemize}

Let $T$ be the set of vertices such that adding edges from all the nodes in
$T$ to node $w$ minimizes the cost of the Nash equilibrium.  Denote by
$G+T$ the graph constructed by adding to $G$ edges from all nodes in
$T$ to $w$. Observe that $T$ cannot contain any nodes with internal
opinion of $0$ as adding an edge from a node with internal opinion $0$
to $w$ which has internal opinion $-3$ can only increase the cost of
the Nash equilibrium. Thus, $T$ contains only vertices of type
$v_{i,j}$ and $u_i$. In the table in Figure \ref{fig_cost} we compute
$v_{i,j}$'s opinion in the Nash equilibrium and the cost of its star
as a function of which vertices that influence $v_{i,j}$ are in $T$.
For example in the first row we consider the case in which
$v_{i,j},u_i,u_j \notin T$. In this case, $v_{i,j}$'s opinion is
$(1+1+1)/3=1$ and the cost of its star is $\frac{1}{2}  \cdot 24 =
12$. We use the costs in this table to reason about the structure of
$T$ and the cost of the Nash equilibrium in $G+T$.

\begin{figure}[!h] 
\begin{center}
\begin{tabular}{|c|c|c|c|}
\hline  & Configuration & $v_{i,j}$'s opinion & $v_{i,j}$'s star cost 
\\\hline 1 & $v_{i,j},u_i,u_j \notin T$ & $1$ & $12$ 
\\\hline 2 & $v_{i,j} \in T,u_i,u_j \notin T$ & $0$ & $12$ 
\\\hline 3 & $v_{i,j},u_i \in T, u_j \notin T $ & $-1/2$ & $14$ 
\\\hline 4 & $v_{i,j},u_i,u_j \in T$ & $-1$ & $20$ 
\\\hline 5 & $v_{i,j},u_j \notin T,u_i \in T$ & $1/3$ & $4$ 
\\\hline 6 & $v_{i,j} \notin T,u_i,u_j \in T$ & $-1/3$ & $4$ 
\\\hline \end{tabular}
 \caption{The total cost of $v_{i,j}$'s star for different configurations}
 \label{fig_cost}
 \end{center}
\end{figure}


In Lemma \ref{to-eq} we show how to construct from $T$ a set $T'$ such that
$\co_{G+T'}(\no') = \co_{G+T}(\no)$ and $T'$ is a \emph {pseudo vertex cover}. 
We say that
a set $T'$ is a pseudo vertex cover if it obeys two properties: first, it contains 
only vertices of the type $u_i$. Second, the vertices in $V'$ corresponding to the $u_i$'s in $V$
constitute a vertex cover in $G'$. 

Next, we consider the cost of the Nash equilibrium in the graph $G+S$ where $S$ is a pseudo vertex cover:
By the table in Figure \ref{fig_cost} we have the cost associated with every $v_{i,j}$'s star is $4$. This 
is by the fact that $S$ is a pseudo vertex cover and hence the only applicable cases are $5$ and $6$,
in both cases the total cost of $v_{i,j}$'s star is $4$. Also, note that the cost  for each $u_i \in S$ is $8$.
Hence, the total cost of the Nash equilibrium for network $G+S$ is $f(S)=4|E|+8|S|$. By construction, $T'$
is a pseudo vertex cover and it also minimizes $f(\cdot)$, since $\co_{G+T}(\no) =  \co_{G+T'}(\no') =4|E|+8|T'|$ 
and $T$ is optimal. Therefore $T'$ corresponds to a minimum vertex cover in $G'$. A key element in this reduction is the property that 
the cost of $v_{i,j}$'s star is the same, whether $u_i\in T'$ or both $u_i$ and $u_j$ belong to $T'$.

\begin{lemma} \label{to-eq}
There exists a pseudo vertex cover $T'$ such that $\co_{G+T'}(\no') = \co_{G+T}(\no)$
\end{lemma}
\begin{proof}
First, we obtain $T''$ by removing from $T$ all vertices of type $v_{i,j}$. We have that $\co_{G+T''}(\no'') \leq \co_{G+T}(\no)$
since by examining the table in Figure \ref{fig_cost} we observe that including vertices of type $v_{i,j}$ in $T''$
can only increase the cost of the Nash equilibrium. Since $T$ is optimal, it has to be the case that $\co_{G+T''}(\no'') = \co_{G+T}(\no)$.
Next, to get $T'$ we take $T''$ and for each vertex $v_{i,j}$ such that $u_i,u_j \notin T''$ we add
$u_i$ to $T'$. By adding these vertices we have not increased the cost since in the worst case
$v_{i,j}$'s star and $u_i$ have a total cost of $12$ which is the same as their previous total cost.
As before by the optimality of $T$ we could not have reduced the cost by adding the vertices, therefore it 
still holds that $\co_{G+T'}(\no') = \co_{G+T}(\no)$. To complete the proof observe that by construction $T'$ is a pseudo vertex cover.
\end{proof}

\end{proof}


\xhdr{Adding an arbitrary set of edges}
In the last case we consider, which is the most general one, we can add any set of edges. 
For this case we leave open the question of the hardness of adding an unrestricted set of edges.
We do show that finding the best set of $k$ arbitrary edges is NP-hard. This is done
by a reduction from $k$-dense subgraph \cite{FeigePK01} :
\begin{proposition} \label{prop:arb-budget-np}
Finding a best set of arbitrary $k$ edges is NP-hard.
\end{proposition}
\begin{proof}
We show a reduction from the ``Dense $k$-Subgraph Problem'' defined in \cite{FeigePK01}:
given an undirected graph $G'=(V',E')$ and a parameter $k$, find a set of $k$ vertices with maximum
average degree in the subgraph induced by this set.
Given an instance of the ``Dense $k$-Subgraph Problem'' we create an instance of
the opinion game as follows:
\begin{itemize}
\item For every edge $(i,j)\in E'$ we create a node $v_{i,j}$ with internal opinion 0.
\item For every vertex $i\in V'$ we create a node $u_i$ with internal opinion 1.
\item For every $v_{i,j}$ we add directed edges $(v_{i,j},u_i)$ and $(v_{i,j},u_j)$.
\item For every $u_i$ we create an in-directed star with $20$ peripheral nodes that have an internal opinion of $0$.
\item Finally, we create a single isolated vertex $w$ with internal opinion -1.
\end{itemize}
The proof is composed of two lemmas. In Lemma \ref{arb_edges_budget1} we show that all edges in the minimizing set are of type $(u_i,w)$. Then we denote by $T$ the set of 
nodes of type $u_i$ such that adding an edge from each one of these nodes to $w$ minimizes the cost, and in Lemma \ref{arb_edges_budget2} we show that $T$ is a $k$ 
densest subgraph.
\begin{lemma} \label{arb_edges_budget1}
The best set of edges to add contains only edges from nodes of type $u_i$ to $w$.
\end{lemma}
\begin{proof}
Our first observation is that any edge which is not from nodes of type $u_i$ affects the cost of at most one node. This is simply because all nodes in the graph, if affected by any 
node at all, are affected by nodes of type $u_i$. The cost of each one of the nodes in the graph in the Nash equilibrium is at most $1$, and therefore the improvement in the cost from 
adding any such edge is at most $1$. On the other hand, adding an edge from nodes of type $u_i$ to $w$ reduces the cost by at least $\frac{1}{2}20-2(1-0)^2=8$. It is easy to 
verify that adding edges from nodes of type $u_i$ to other nodes has a smaller effect on $i$'s cost.
\end{proof}
\begin{lemma}  \label{arb_edges_budget2}
The previously defined set $T$ is a solution to the dense $k$-subgraph problem.
\end{lemma}
\begin{proof}
The key point is the fact that the cost associated with a node of type
$v_{i,j}$ is $0$ if and only if both $u_i$ and $u_j$ are in $T$;
otherwise this cost is exactly $\frac{2}{3}$.
When $u_i \in T$, the opinion of $u_i$ in the Nash equilibrium is $0$
since it is averaging between $1$ and $-1$. Therefore node
$v_{i,j}$'s associated cost in the Nash equilibrium is:
\begin{itemize}
\item $0$ - if both $u_i$ and $u_j$ are in $T$ - since $v_{i,j}$ holds opinion $0$.
\item $\frac{2}{3}$ -  if both $u_i$ and $u_j$ are not in $T$ - since $v_{i,j}$'s opinion is $\frac{2}{3}$ and therefore the cost is $(0-\frac{2}{3})^2+2(1-\frac{2}{3})^2=\frac{2}{3}$.
\item $\frac{2}{3}$ -  if only one of $u_i$, $u_j$ is in $T$ - then $v_{i,j}$'s opinion is $\frac{1}{3}$ and therefore the cost is $(0-\frac{1}{3})^2+(0-\frac{1}{3})^2+(1-\frac{1}{3})^2=\frac
{2}{3} $.
\end{itemize}
Hence to minimize the cost of the Nash equilibrium we should choose a set $T$ maximizing the number of nodes of type $v_{i,j}$ for which both $u_i$ and $u_j$ are in $T$. In the 
graph $G'$ from the $k$-dense subgraph problem that set $T$ is a set of vertices and what we are looking for is the set $T$ with an induced graph that has 
the maximum number of edges. 
By definition this set is exactly a $k$-densest subgraph.
\end{proof}
\end{proof}

Finding approximation algorithms for all of the problems discussed
in propositions~\ref{prop-from-is-np}, \ref{prop-to-is-np}, 
and \ref{prop:arb-budget-np}
is an interesting question. As a first step we offer a
$\frac{9}{4}$-approximation for the problem of optimally adding
edges to a directed graph $G$ --- a problem whose hardness for
exact optimization we do not know.
The approximation algorithm works simply by including
the reverse copy of every edge in $G$ that is not already in $G$;
this produces a bi-directed graph $G'$.
\begin{claim} \label{claim:arb-edge-approx}
$c_{G'}(\no') \leq \frac{9}{4} c_G(\oo) $.
\end{claim}
\begin{proof}
By Theorem \ref{PoAlongProof} we have that $\co_{G'}(\no') \leq
\frac{9}{8} \co_{G'}(\oo')$. Also notice that in the worst case, in
order to get from $G$ to $G'$, we must double all the edges in $G$.
Therefore $ \co_{G'}(\oo') \leq 2\co_{G}(\oo)$. By combining the two
we have that $ \co_{G'}(\no') \leq \frac{9}{4} \co_{G}(\oo)$.
\end{proof}

For weighted graphs we can also include reverse copies of edges that do appear in $G$ and hence
achieve an approximation ratio of $2$ for analogous reasons.

\subsection{Adding a Single Weighted Edge}

We now consider how to optimally choose the weight to put on a
single edge $(i,j)$, to minimize the cost of the Nash equilibrium.
Suppose we add weight $\rho$ to the edge $(i,j)$.
The modified Laplacian is
$\displaystyle{ L' = L + \rho e_i (e_i-e_j)^T,}$
where $e_i$ is the $i$th vector in the standard basis.
The modified Nash equilibrium is 
$\displaystyle{  \no' = (L'+I)^{-1} \io = ( (L+I) + \rho e_i (e_i-e_j)^T )^{-1} \io.}$

Using the Sherman-Morrison formula for the rank-one update to an 
inverse~\cite[\S2.1.3]{GVL}, we have
\begin{align*}
  \no' &= 
    \left[ (L+I)^{-1} - 
           \frac{ (L+I)^{-1} \rho e_i (e_i - e_j)^T (L+I)^{-1} } 
                { 1 + \rho (e_i-e_j)^T (L+I)^{-1} e_i }
    \right] \io \\
  &= 
    \no -
    \oi_{i}
    \left( \frac{\rho(\no_i-\no_j)}{1+\rho (\oi_{i,i}-\oi_{i,j})} \right),
\end{align*}
where $\oi_i = (L+I)^{-1} e_i$ is the influence of $\io_i$ on the Nash opinions
in the original graph.  
Therefore, $\oi_i$ gives the direction of change
of the Nash equilibrium when the weight on $(i,j)$ is increased:
the equilibrium opinions all shift in the direction of $\oi_i$.
We prove the following key properties of this influence
vector $\oi_i$:

\begin{lemma}
The entries of $\oi_i = (L+I)^{-1} e_i$ lie in $[0,1]$,
and $\oi_{i,i}$ is the unique maximum entry.
\label{lem:inf-vector-properties}
\end{lemma}
\begin{proof}
The influence vector $\oi_i$ is simply the Nash equilibrium for
  the internal opinion vector $e_i$.  The Nash equilibrium
  is the limit of repeated averaging starting from the internal opinions,
  and the average of numbers in $[0,1]$ is in $[0,1]$.  Thus
  the entries of $\oi_i$ are in $[0,1]$.

  We show that $\oi_{i,i}$ is the maximal entry by contradiction.  Suppose
  $\oi_{i,j}$ is maximal for some $j \neq i$.  Because $L+I$ is nonsingular,
  $\oi_i$ cannot be the zero vector, so $\oi_{i,j} > 0$.  The
  equilibrium equations for $j$ can be written
\begin{align*}
    \oi_{i,j} &= \frac{\sum_{k \in N(j)} w_{j,k} \oi_{i,k}}
                    {1+\sum_{k \in N(j)} w_{j,k}} \\
            & \leq \left( \frac{\sum_{k \in N(j)} w_{j,k}}
                              {1+\sum_{k \in N(j)} w_{j,k}} \right)
                  \max_{k \in N(j)} \oi_{i,k}
             \leq \max_{k \in N(j)} \oi_{i,k}
\end{align*}
  where the final inequality is strict if $\oi_{i,k} \neq 0$ for any $k \in N(j)$.
  But 
  $\oi_{i,k} \neq 0$
  for some $k \in N(j)$, since otherwise
  $\oi_{i,j}$ would be zero.  Therefore, there must be some $k \in N(j)$
  such that $\oi_{i,k} > \oi_{i,j}$, which contradicts the hypothesis that
  $\oi_{i,j}$ is maximal.
\end{proof}

We now show how to choose the {\em optimal} weight $\rho$ to add to edge $(i,j)$
to best reduce the social cost of the Nash equilbrium.
\begin{theorem}
The optimal weight $\rho$ to add to the edge $(i,j)$ can be
computed in polynomial time.
\label{thm:opt-weight-add}
\end{theorem}
\begin{proof}
Note that
\begin{align*}
   \no_i'-\no_j' &= 
   \left( \no_i-\no_j \right) 
   \left( 1-\frac{\rho (\oi_{i,i}-\oi_{i,j})}{1+\rho (\oi_{i,i}-\oi_{i,j})} \right) \\
   &=\frac{\no_i-\no_j}{1+\rho(\oi_{i,i}-\oi_{i,j})},
\end{align*}
and we can write the new Nash equilibrium as $\no' = \no - \phi \oi_i$, where
\[
  \phi = \frac{\rho (\no_i-\no_j)}{1+\rho(\oi_{i,i}-\oi_{i,j})}
       = \rho (\no_i-\no_j) \frac{\no_i'-\no_j'}{\no_i-\no_j}
       = \rho (\no_i'-\no_j').
\]
For small values of $\rho$, we have that 
$\phi = \rho (\no_i-\no_j) + O(\rho^2)$; 
and as $\rho \rightarrow \infty$, we have that 
$\phi \rightarrow \phi_{\max} = (\no_i-\no_j)/(\oi_{i,i}-\oi_{i,j})$ and
$\no_i'-\no_j' \rightarrow 0$.  Thus, adding a small amount of weight
to edge $(i,j)$ moves the Nash equilibrium in the direction of the
influence vector $\oi_i$ proportional to the weight $\rho$ and the
discrepancy $\no_i-\no_j$; while adding larger amounts of weight moves
the Nash equilibrium by a bounded amount in the direction of the
influence vector $\oi_i$, with the asymptotic limit of large edge
weight corresponding to the case when $i$ and $j$ have the same
opinion.

What does adding a weighted edge between $i$ and $j$ do to the social
cost at Nash equilibrium?  In the modified graph, the social cost is
\[
  c'(\ao) = \ao^T A \ao + \rho (\ao_i-\ao_j)^2 + \|\ao-\io\|^2.
\]
At the new Nash equilibrium, we have
\begin{align*}
  c'(\no') 
  &= \no'^T A \no' + \rho (\no_i'-\no_j')^2 + \|\no'-\io\|^2 \\
  &= \no'^T A \no' + \phi (\no_i'-\no_j') + \|\no'-\io\|^2.
\end{align*}
Because $\no'$ is a linear function of $\phi$, the above shows that
$c'(\no')$ is a quadratic function of $\phi$, which we can simplify to
$c'(\no') = \alpha_{ij} \phi^2 - 2\beta_{ij} \phi + c(\no),$
where
\begin{align*}
  \alpha_{ij} &= \oi_i^T (A+I) \oi_i -(\oi_{i,i}-\oi_{i,j}) \\
  \beta_{ij}  &= \oi_i^T \left( (A+I) \no - \io \right) - 
                \frac{1}{2}(\no_i-\no_j).
\end{align*}
The range of possible values for $\phi$ is between $0$ (corresponding to
$\rho = 0$) and $\phi_{\max}$ (corresponding to the limit as $\rho$ goes
to infinity).  Subject to the constraints on the range of $\phi$, 
the quadratic in $\phi$ is minimal either at $0$, 
at $\phi_{\max}$, or at $\beta_{ij}/\alpha_{ij}$ (assuming this point
is between $0$ and $\phi_{\max}$).
We can therefore determine the optimal weight for a single edge
in polynomial time.
\end{proof}

Note that the above computations also give us a simple formula for
the gradient components $\gamma_{ij}$ corresponding to differentiation
with respect to $w_{ij}$:
\begin{align*}
  \gamma_{ij} &\equiv \frac{d[c'(\no')]}{d\rho} 
  = \frac{d[c'(\no')]}{d\phi} \frac{d\phi}{d\rho}
  = -2\beta_{ij} (\no_i-\no_j) \\
  &= (\no_i-\no_j)^2 - 2 (\no_i-\no_j) \oi_i^T ((A+I) \no - \io).
\end{align*}
The residual vector $(A+I) \no - \io$ measures the extent to which $\no$ fails
to satisfy the equation for the socially optimal opinion $\oo$.  If
this vector is large enough, and if the influence vector $v_i$ is
sufficiently well aligned with the residual, then adding weight to 
the $(i,j)$ edge can decrease the social cost at Nash equilibrium.
Thus, though computing a globally optimal choice of additional edge weights
may be NP-hard, we can generally compute locally optimal edge additions
via the method of steepest descent.

\xhdr{Acknowledgments}
We thank 
Michael Macy for valuable discussion.

\bibliographystyle{IEEEtranS}


\end{document}